\documentclass[10pt, final,leqno]{siamltex}

\usepackage{float}
\usepackage{subfigure}
\usepackage{lineno}
\usepackage{graphicx}
\usepackage{color}
\usepackage{amsmath}
\usepackage{enumerate}

\renewcommand{\vec}[1]{\ensuremath\boldsymbol{#1}}
\newcommand{\p}{\ensuremath\partial}

\title{The Three-Dimensional Jump Conditions for the Stokes Equations with Discontinuous Viscosity, Singular Forces, and an Incompressible Interface
	\thanks{This work was supported by the University at Buffalo SUNY and NSF Award \#1253739}}

\author{Prerna Gera \and David Salac \thanks{Department of Mechanical and Aerospace Engineering, University at Buffalo SUNY, Buffalo, NY, 14260.}}

\bibstyle{siam}

\begin{document}

\maketitle

\begin{abstract}
	The three-dimensional jump conditions for the pressure and velocity fields, up to the second normal derivative,
	across an incompressible/inextensible interface	in the Stokes regime are derived herein. 
	The fluid viscosity is only piecewise continuous in the domain while the embedded interface exerts 
	singular forces on the surround fluids. This gives rise to discontinuous solutions in the pressure and 
	velocity field. These jump conditions are required to develop accurate numerical methods, such 
	as the Immersed Interface Method, for the solutions of the Stokes equations in such situations.

\end{abstract}
\begin{keywords} 
	Stokes equations, discontinuous viscosity, singular force, immersed interface method, jump conditions, incompressible interface, inextensible interface
\end{keywords}

\begin{AMS}
	65N06, 65N12, 76Z05, 76D07, 35R05
\end{AMS}

\section{\label{sec:introduction} Introduction}

Biological system containing cells can be considered as a multiphase fluid
problem, with the cell membranes forming a boundary between fluids of differing material properties.
Unlike standard multiphase fluid problems, such as droplets or bubbles, biological 
multiphase systems, such as red blood cells surrounded by blood plasma, are characterized by 
inextensible membranes and singular forces due to both a surface-tension like contribution and a resistance to bending \cite{Pozrikidis2003}.
Various techniques have been developed to model such systems. 
These techniques can be split into two categories, based on how they track the location of the interface.
Lagrangian methods, such as the boundary integral method \cite{Veerapaneni2009,Veerapaneni2011}
explicitly track the location of the interface. While these method can be highly accurate, they 
are difficult to extend to three-dimensions and are not well suited when the membrane undergoes
large deformations. Eulerian method, on the other hand, implicitly track the location
of the interface. Examples of Eulerian methods include the phase-field method \cite{Biben2005,Lowengrub2009}
and the level set method \cite{Salac2011,Salac2012b}. Advantages of the Eulerian methods include 
easy extensibility to three dimensions and the ability to handle large membrane deformations naturally.
The major disadvantage is the added computational cost of implicit membrane tracking. 

Eulerian methods can be further divided into two categories depending on how they treat
the singular forces arising at the embedded membrane. The first type, characterized 
by the immersed boundary method \cite{Peskin1977} or the continuum surface force method \cite{Chang1996},
distribute (smooth) the singular forces over a small region near the interface. This essentially
turns the membrane forces into localized body-force terms, which are then included in the 
fluid governing equations. For most situations this smoothing results in a first-order accurate 
method for the fluid equations \cite{Tan2009}. 

The second type of Eulerian method avoids the smearing of the interface by explicitly incorporating
the singular forces and embedded boundary conditions in the field equations. These forces and immersed boundary conditions
are included by explicitly considering the jump in the solution and derivatives of the solution across the interface. 
Mayo \cite{Mayo1984} used this idea to solve the Poisson and Biharmonic equation on irregular domains.
Leveque and Li called this technique the Immersed Interface Method (IIM) and used it to solve Elliptic equations \cite{Leveque1994}.
They later extended this to Stokes flow with elastic boundaries or surface tension \cite{Leveque1997}. Later the immersed interface
method was further developed to model the Navier-Stokes equations \cite{Le2006109,Li2001822,Xu2006454}.
Interested readers are referred to the book by Li and Ito for more information about the immersed interface method \cite{li2006immersed}.

The application of the immersed interface method to a multiphase fluid system requires that the jumps in the velocity and pressure field
be explicitly calculated. For standard multiphase systems these jump conditions for a discontinuous viscosity across the interface have been
determined both in two dimensions \cite{Tan20089955,Tan2009} and three dimensions \cite{Ito2006,Xu2006}. Only recently has work been
done on extending the immersed interface method to multiphase flows with inextensible membranes \cite{Tan2012}. This recent work, though,
is limited to the two-dimensional, constant-viscosity case and is not applicable to the more general discontinuous viscosity situation.

In this work the three-dimensional jump conditions for a multiphase Stokes flow system with an inextensible interface, 
singular forces, and a discontinuous viscosity are for the first time derived. The jump up to the 
second normal derivative of both the pressure and velocity fields will be developed.
These jump conditions will be used to construct an Immersed Interface solver. Several analytic test cases 
are also developed to verify the accuracy of the jump conditions.

This paper is organized as follows. First, the basic governing equations
for the three-dimensional Stokes system with an inextensible membrane are presented.
Next, in Sec. \ref{sec:IIM} the immersed interface method is briefly outlined. The jump conditions for 
the velocity and pressure in the Stokes equations are developed in Sec. \ref{sec:JumpConditions}. 
The derivation of analytic test cases is shown in Sec. \ref{sec:TestCases}
while results and convergence analysis follow in Sec. \ref{sec:Results}. A brief conclusion and future work follows
in Sec. \ref{sec:Conclusion}.

\section{\label{sec:StokesEquations} Governing Equations}

Let a three-dimensional bounded domain, $\Omega$, contain a closed and incompressible material interface $\Gamma$.
For simplicity assume that the interface is wholly contained in the domain.
The region enclosed by the interface is denoted as $\Omega^-$ while the region outside 
the interface is denoted by $\Omega^+$, Fig. \ref{fig:CompDomain}. The fluid in each domain is 
modeled using the Stokes equations:
\begin{align}
	\label{eq:momentum}	\nabla p^\pm =& \mu^\pm \nabla^2\vec{u}^\pm+\vec{g}^\pm, \\
	\label{eq:divfree}	\nabla\cdot\vec{u}^\pm=&0,
\end{align}
with boundary conditions $\vec{u}^+=\vec{u}_b$ on the boundary of $\Omega$. The pressure, $p$,
body force term, $\vec{g}$, and viscosity, $\mu$, are all piecewise constant in the domain, with a finite jump
occurring across the interface. The velocity field, $\vec{u}=(u,v,w)$, is assumed to be $C^0-$continuous across the 
interface. It is also assumed that the velocity obey a no-slip condition on the interface.
It is thus possible to state that the velocity of the interface $\Gamma$ is given by $\vec{q}=\vec{u}^+=\vec{u}^-$.

Application of the volume incompressibility constraint, Eq. (\ref{eq:divfree}), to the momentum equations, Eq. (\ref{eq:momentum}),
leads to a pressure-Poisson equation in each domain,
\begin{equation}
	\label{eq:pressue-poisson}
	\nabla^2 p^\pm=\nabla\cdot\vec{g}^\pm,
\end{equation}
with a boundary condition on the outer domain of
\begin{equation}
	\frac{\partial p^+}{\partial \vec{n}_b}=\mu^+ \vec{n}_b\cdot\nabla^2\vec{u}^+ +\vec{n}_b\cdot\vec{g}^+,
\end{equation}
where $\vec{n}_b$ is the outward facing normal to the domain $\Omega$.

\begin{figure}[ht]
	\center
	\includegraphics[width=2in]{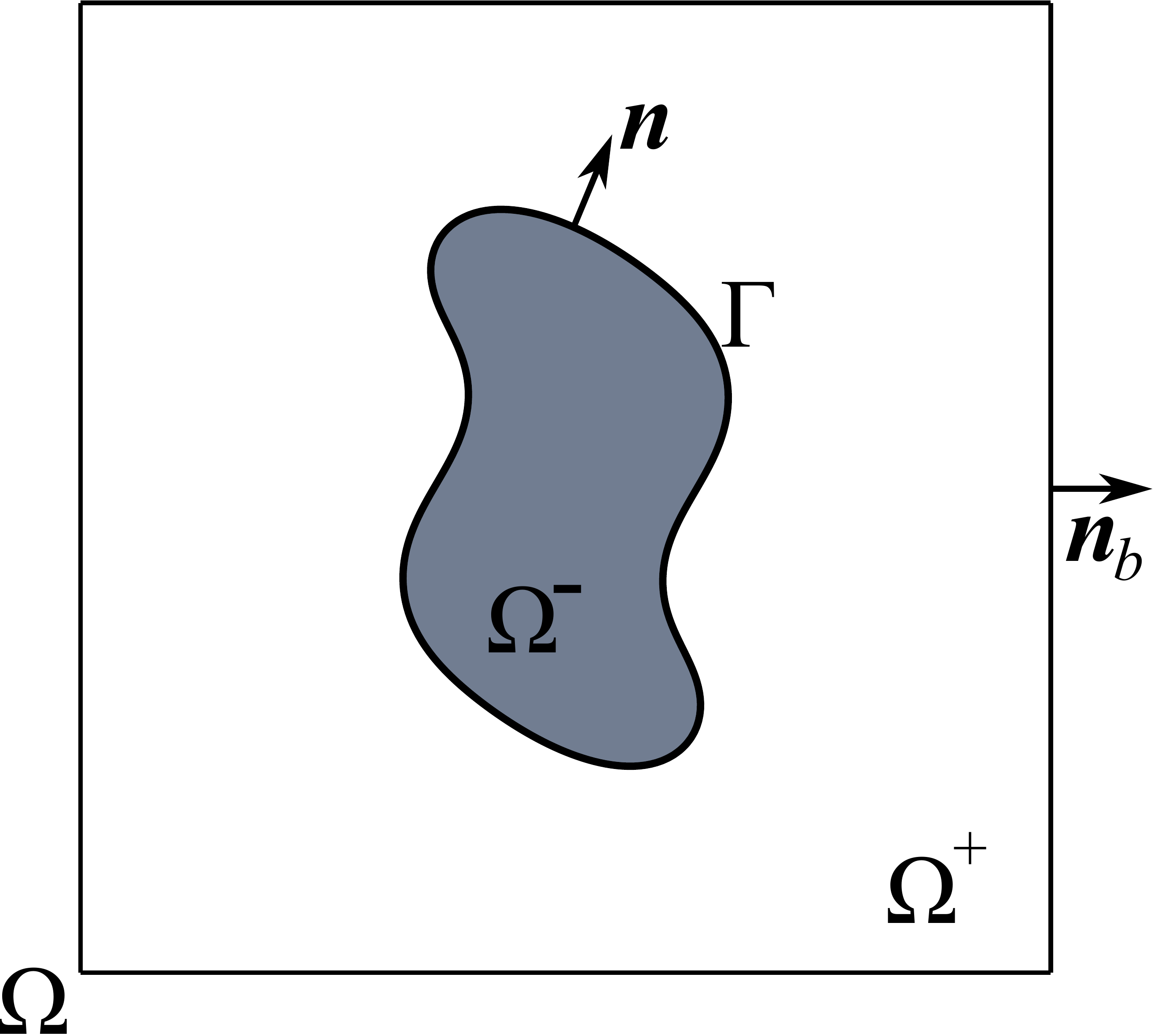}
	\caption{\label{fig:CompDomain} Sample computational domain.}
\end{figure}	

The incompressible interface requires that the velocity on the interface be divergence-free:
\begin{equation}
	\label{eq:surface_constraint}
	\nabla_s\cdot\vec{q}=0,
\end{equation}
where $\nabla_s=\left(\vec{I}-\vec{n}\otimes\vec{n}\right)\nabla$ is the surface gradient and $\nabla_s\cdot$ is the surface divergence operator.

The presence of the interface results in a singular force, $\vec{f}$, exerted on the surrounding fluids.
This singular force is application specific. In vesicle/cell simulations
this force would be composed of a variable surface tension, $\gamma$, and other singular forces, $\vec{f}_0$:
\begin{equation}
	\label{eq:singular_forces}
	\vec{f}=-\gamma H \vec{n}+\nabla_s\gamma+\vec{f}_0,
\end{equation}
where $\vec{n}$ is the outward facing normal to the interface $\Gamma$ (from $\Omega^-$ into $\Omega^+$) and $H$ is the total curvature. An example of 
other singular forces which might be applied are the bending forces in vesicle and red-blood cell simulations \cite{Salac2012b,Seifert1997} while
the surface tension will be chosen to enforce the surface incompressibility constraint, Eq. (\ref{eq:surface_constraint}). Future work will consider
singular forces of this form.

The singular force is balanced by a jump in the normal component of the stress tensor:
\begin{equation}
	\label{eq:stress_balance}
	\left(\vec{\sigma}^+-\vec{\sigma}^-\right)\cdot\vec{n}=\vec{f},
\end{equation}
where the stress tensor is given by
\begin{equation}
	\vec{\sigma}^\pm=-p^\pm\vec{I}+\mu^\pm\left(\nabla\vec{u}^\pm+\left(\nabla\vec{u}^\pm\right)^T\right).
\end{equation}

\section{\label{sec:IIM} Immersed Interface Method}

The governing equations as described in Sec. \ref{sec:StokesEquations} result in two coupled, but discontinuous, fluid fields. In this section
the scheme developed in Ref. \cite{Russell2003} and used to solve the two-dimensional Stokes equations in Ref. \cite{Lai200899} is reviewed. 
This is also the impetus for deriving the normal jump conditions across the interface.

Begin by 
discritizing the domain $\Omega$ using a uniform Cartesian mesh with a grid spacing of $h$ in all directions. 
While this section will focus on two-dimensional domains 
the extension to three-dimensions is straight-forward and used in the subsequent sections.
Now consider the solution to a generic Poisson problem of the
form 
\begin{equation}
	\label{eq:generic_poisson}
	\nabla^2 u=f,
\end{equation}
where both $u$ and $f$ are (possibly) discontinuous across the interface. Upon discritizing Eq. (\ref{eq:generic_poisson}) it is possible to classify all grid
points as either regular or irregular, Fig. \ref{fig:IIMStencils}. Regular nodes are those nodes where the interface does not cross the discritized stencil. 
As the solution is continuous away from the interface the discretization at regular nodes does not need to be modified. 
Irregular nodes are defined
as those where the interface does cross the discritized stencil. As the solution $u$ may be discontinuous across the interface it is not possible to use
the standard discretization. Instead the system is modified at irregular points to take into account the discontinuity.

\begin{figure}[ht]
	\center
	\includegraphics[width=2.0in]{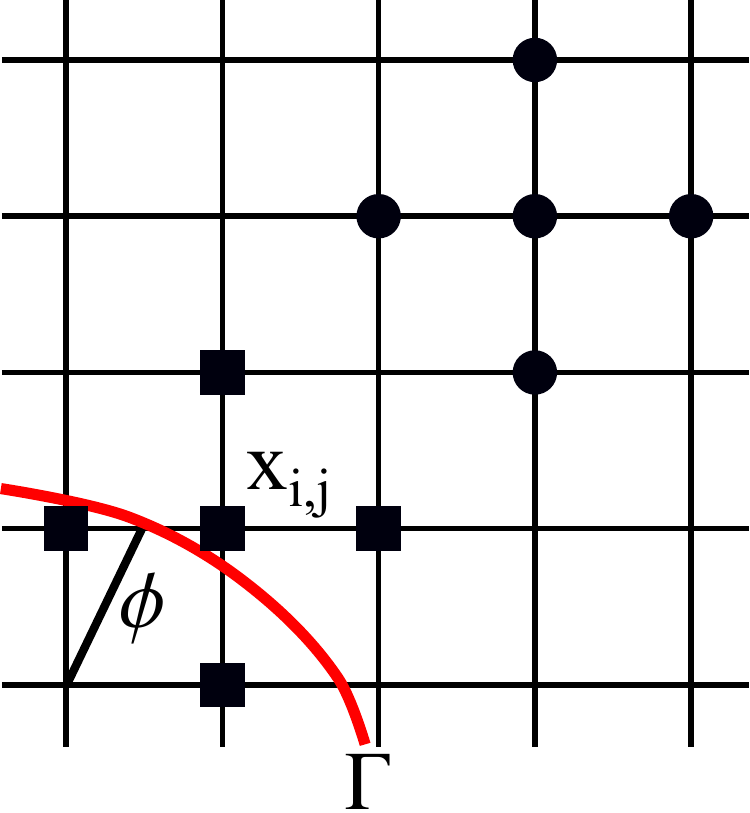}
	\caption{\label{fig:IIMStencils} Discritized Laplacian at a regular node (circle) and irregular node (square). Irregular nodes are those where the interface 
	intersects the stencil. The signed distance function from any grid point to the interface is given by $\phi$.}
\end{figure}	

Consider the irregular point $\vec{x}_{i,j}$ shown in Fig. \ref{fig:IIMStencils}. Let this point exist in the $\Omega^+$ domain.
A standard second-order finite difference discretization results in
\begin{equation}
	\label{eq:poisson_discritization}
	\frac{u^+_{i-1,j}+u^+_{i+1,j}+u^+_{i,j-1}+u^+_{i,j+1}-4 u^+_{i,j}}{h^2}=f^+_{i,j}.
\end{equation}
The issue is the values at $\vec{x}_{i-1,j}$ and $\vec{x}_{i,j-1}$ exist in the $\Omega^-$ domain, not the $\Omega^+$ domain. 
This fact must be taken into account in the discretization.

Let a jump across the interface be defined as 
\begin{equation}
	[u]=\lim_{\epsilon\to 0^+} u(\vec{x}_\Gamma+\epsilon \vec{n}) - \lim_{\epsilon\to 0^+} u(\vec{x}_\Gamma-\epsilon \vec{n}),
\end{equation}
where $\vec{x}_\Gamma$ is a point on the interface and $\vec{n}$ is the outward normal vector. Assume that the jumps in the solution, $[u]_\Gamma$,
the first normal derivative, $[\partial u/\partial \vec{n}]_\Gamma$, and the second normal derivative $[\partial^2 u/\partial \vec{n}^2]_\Gamma$ 
are known on the interface for the problem given in Eq. (\ref{eq:generic_poisson}). 
Using a Taylor Series expansion in the normal direction about an interface point the jumps may be extended to a grid point by
\begin{equation}
	[u]_{i,j} = [u]_\Gamma+\phi_{i,j} \left[\frac{\partial u}{\partial \vec{n}}\right]_\Gamma+\frac{\phi_{i,j}^2}{2} \left[\frac{\partial^2 u}{\partial \vec{n}^2}\right]_\Gamma
		+\mathcal{O}(h^3),
\end{equation}
where $\phi_{i,j}$ is the signed-distance from the grid point $\vec{x}_{i,j}$ to the interface, Fig. \ref{fig:IIMStencils}. Note that by using the 
signed distance functions it is possible to account for extension of the jumps into either domain.

By extending the solution jumps from the interface to the grid points it can be written that $u^+_{i,j}=u^-_{i,j}+[u]_{i,j}$.
Use this expression in the Eq. (\ref{eq:poisson_discritization}) to get
\begin{align}
	\label{eq:poisson_discritization1}
	\frac{u^-_{i-1,j}+u^+_{i+1,j}+u^-_{i,j-1}+u^+_{i,j+1}-4 u^+_{i,j}}{h^2} \\ \nonumber
		+\frac{[u]_{i-1,j}+[u]_{i,j-1}}{h^2}=f^+_{i,j}.
\end{align}
As the jumps $[u]_{i-1,j}$ and $[u]_{i,j-1}$ are known they can be moved to the right-hand side as
explicit corrections on the linear system,
\begin{align}
	\label{eq:poisson_discritization2}
	\frac{u^-_{i-1,j}+u^+_{i+1,j}+u^-_{i,j-1}+u^+_{i,j+1}-4 u^+_{i,j}}{h^2}=f^+_{i,j}+C_{i,j},
\end{align}
where $C_{i,j}$ is the total correction needed at an irregular node at location $\vec{x}_{i,j}$. In this case
$C_{i,j}=-([u]_{i-1,j}+[u]_{i,j-1})/h^2$. This method is easily extended to irregular nodes on either side of the interface and 
to three-dimensional systems.

It should be noted that the extension of the jumps must be calculated up to third-order accuracy to ensure that irregular nodes have
a local truncation error of $\mathcal{O}(h)$. Despite this higher local truncation error the overall method will retain the second-order
accuracy of the underlying discretization \cite{Adams2002, Leveque1997, Li2001}. If the second-normal derivative is not taken into 
account the local truncation error for irregular nodes will be reduced to $\mathcal{O}(1)$ and the overall scheme would only be first-order.

\section{\label{sec:JumpConditions} Derivation of Jump Conditions}

The governing equations in Sec. \ref{sec:StokesEquations} result in two coupled Poisson problems. 
Using the Immersed Interface Method of Sec. \ref{sec:IIM} the two systems can be written over the entire domain by including the appropriate corrections:
\begin{align}
	\label{eq:augmented_momentum} \nabla^2 \left(\mu\vec{u}\right)=&\nabla p-\vec{g}+\vec{C}_u,\\
	\label{eq:augmented_pressure} \nabla^2 p=&\nabla\cdot\vec{g}+C_p,
\end{align}
where the corrections $\vec{C}_u$ and $C_p$ are non-zero only at irregular grid points. To calculate these corrections at irregular grid points the jumps up to the second-normal derivatives
need to be derived for both the pressure and the augmented velocity, $\mu \vec{u}$. This section will outline the derivation of these jump conditions. It should be noted that
a $C^0$-continuous velocity indicates that $[\vec{u}]=0$. It should also be noted that the interface can be approached from either the $\Omega^-$ or $\Omega^+$ domains.
It is thus defined that 
\begin{align}
	f^-=&\lim_{\epsilon\to 0^+} f(\vec{x}_\Gamma+\epsilon \vec{n}) \\
	f^+=&\lim_{\epsilon\to 0^+} f(\vec{x}_\Gamma-\epsilon \vec{n}),
\end{align}
where $f$ is a quantity, such as velocity, of interest.

Let $\vec{n}$ be the outward (pointing into $\Omega^+$) unit normal while $\vec{t}$ and $\vec{b}$ are two unit tangent vectors chosen such that $\left(\vec{n},\vec{t},\vec{b}\right)$
form a Darboux Frame. Also chose $\vec{t}$ and $\vec{b}$ to lie along the principle directions of the 
interface, Fig. \ref{fig:LocalVectors}. Denote the directions in the normal and two tangent vectors as $n$, $t$, and $b$, respectively. In the following
sections all derivative in a particular direction are denoted as $\partial/\partial n$, $\partial/\partial t$, and $\partial/\partial b$. Note that a lack of bold-face 
indicates a direction and not a vector.

Theorems \ref{thm:normal_jump} to \ref{thm:slapdn} are non-trivial and non-obvious relations that will be used to develop the jump conditions
for the pressure and velocity.

\begin{figure}[ht]
	\center
	\includegraphics[width=1.5in]{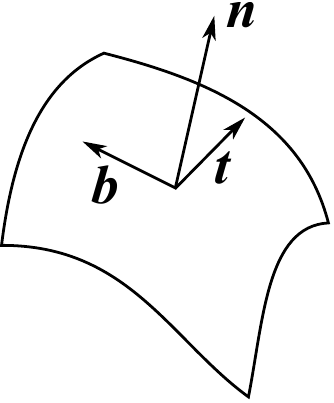}
	\caption{\label{fig:LocalVectors} The local vectors on an interface point. The unit vector $\vec{n}$ is the outward facing normal while $\vec{t}$ and $\vec{b}$ are
	unit tangent vectors which lie along the principle directions at the point on the interface.}
\end{figure}	

\begin{theorem}
	\label{thm:normal_jump}
	The jump in the viscosity times the normal derivative of the velocity in the normal direction is zero:
	\begin{equation}
		\label{eq:normal_derivative}
		\left[\mu \frac{\p \vec{u}}{\p n}\right]\cdot\vec{n}=0
	\end{equation}
\end{theorem}
\begin{proof}
	For the fluid in each domain the incompressibility conditions holds, $\mu^\pm \nabla\cdot\vec{u}^\pm=0$. Thus, the jump in the incompressibility
	condition is zero across the interface. In terms of the local vectors this is 
	\begin{equation}
		\label{eq:volume_incompressible}
		\left[\mu\left(\frac{\p \vec{u}}{\p n }\cdot\vec{n}+\frac{\p \vec{u}}{\p t}\cdot\vec{t}+\frac{\p \vec{u}}{\p b}\cdot\vec{b}\right)\right]=0.
	\end{equation}
	
	The velocity as you approach the interface must also be surface-divergent free, $\mu^\pm \nabla_s\cdot\vec{u}^\pm=0$.
	Therefore the the jump in the surface divergence of the 
	velocity must also be zero,
	\begin{equation}
		\label{eq:surface_incompressible}
		\left[\mu\left(\frac{\p \vec{u}}{\p t}\cdot\vec{t}+\frac{\p \vec{u}}{\p b}\cdot\vec{b}\right)\right]=0.
	\end{equation}	
	Combining Eqs. (\ref{eq:volume_incompressible}) and (\ref{eq:surface_incompressible}) results in 
	\begin{equation}		
		\left[\mu\left(\frac{\p \vec{u}}{\p n}\right)\cdot\vec{n}\right]=0.		
	\end{equation}
	As the normal vector is continuous across the interface, $\vec{n}^-=\vec{n}^+$, the jump condition, Eq. (\ref{eq:normal_derivative}), is obtained.
\qquad \end{proof}

\begin{theorem}
	\label{thm:dudn2}
	The following expression holds on the interface,
	\begin{equation}
		\left[\mu\left(\frac{\p^2 \vec{u}}{\p n^2}\cdot\vec{n}+\frac{\p^2\vec{u}}{\p n\p t}\cdot\vec{t}+\frac{\p^2\vec{u}}{\p n\p b}\cdot\vec{b}\right)\right]=0.
	\end{equation}
\end{theorem}
\begin{proof}
	The incompressibility constraint, $\nabla\cdot\vec{u}=0$, must hold identically in both domains. Thus, the normal derivative of this constraint must
	also be zero, $\p \left(\nabla\cdot\vec{u}\right)/\p n=0$. It is thus possible to write
	\begin{align}
		\left[\mu\frac{\p}{\p n}\left(\nabla\cdot\vec{u}\right)\right]=& \left[\mu \frac{\p}{\p n}\left( \frac{\p\vec{u}}{\p n}\cdot\vec{n} + \frac{\p\vec{u}}{\p t}\cdot\vec{t} + \frac{\p\vec{u}}{\p b}\cdot\vec{b}\right) \right]		\nonumber \\
			=& \left[\mu\left( 	\frac{\p^2\vec{u}}{\p n^2}\cdot\vec{n} + \frac{\p^2\vec{u}}{\p n \p t}\cdot\vec{t} + \frac{\p^2\vec{u}}{\p n \p b}\cdot\vec{b}	\right)\right] \nonumber \\
			=& 0,
	\end{align}
	where the fact that $\left[\p \vec{n}/\p n\right]=0$, $\left[\p \vec{t}/\p n\right]=0$, and $\left[\p \vec{b}/\p n\right]=0$ are held to be true.
\qquad \end{proof}

\begin{theorem}
	\label{thm:slapdn}
	The following expression holds on the interface,
	\begin{equation}
		\label{eq:slapdn}
		\left[\mu \nabla_s^2\vec{u}\right]\cdot\vec{n}=\left[\mu\left(\frac{\p^2\vec{u}}{\p t^2}+\frac{\p^2\vec{u}}{\p b^2}\right)\right]\cdot\vec{n}.
	\end{equation}
\end{theorem}
\begin{proof}
	The full Laplacian in the local frame can be written as	
	\begin{equation}
		\label{eq:split_lap}
		\nabla^2\vec{u}=\frac{\p^2\vec{u}}{\p n^2}+\frac{\p^2\vec{u}}{\p t^2}+\frac{\p^2\vec{u}}{\p b^2},
	\end{equation}
	while the surface Laplacian can be written as \cite{Xu2003}
	\begin{equation}
		\nabla_s^2\vec{u}=\nabla^2\vec{u}-\frac{\p^2\vec{u}}{\p n^2}-H\frac{\p \vec{u}}{\p n}.
	\end{equation}
	Thus the following holds,
	\begin{equation}
		\left[\mu\nabla_s^2\vec{u}\right]\cdot\vec{n}=
					\left[\mu\left(\frac{\p^2\vec{u}}{\p t^2}+\frac{\p^2\vec{u}}{\p b^2}\right)\right]\cdot\vec{n} 
					-H\left[\mu\frac{\p \vec{u}}{\p n}\right]\cdot\vec{n}.
	\end{equation}
	Using the result of Theorem \ref{thm:normal_jump} the expression in Eq. (\ref{eq:slapdn}) is shown to hold true.
	
\qquad \end{proof}

The next three theorems outline the derivation of the pressure jump conditions.

\begin{theorem}
	\label{thm:pressure_jump}
	The jump in the pressure is given by
	\begin{equation}
		\label{eq:pressure_jump}
		[p]=-\vec{f}\cdot\vec{n}.
	\end{equation}
\end{theorem}
\begin{proof}
	The balance of forces on the interface, Eq. (\ref{eq:stress_balance}), is
	\begin{equation}
		\left[-p\vec{n}+\mu\left(\nabla\vec{u}+\left(\nabla\vec{u}\right)^T\right)\cdot\vec{n}\right]=\vec{f}.
	\end{equation}
	Taking the normal component of this balance results in
	\begin{equation}
		[p]=2\left[\mu \frac{\p\vec{u}}{d n}\right]\cdot\vec{n}-\vec{f}\cdot\vec{n}.
	\end{equation}
	Using Eq. (\ref{eq:normal_derivative}) results in the jump of the pressure field, Eq. (\ref{eq:pressure_jump}).
\qquad \end{proof}
\begin{theorem}
	The jump in the normal derivative of the pressure may be written as
	\begin{equation}
		\label{eq:dpdn1}
		\left[\frac{\p p}{\p n}\right]=\left[\mu \nabla^2\vec{u}\right]\cdot\vec{n}
					+ \left[\vec{g}\right]\cdot\vec{n},
	\end{equation}
	or alternatively as
	\begin{align}
		\label{eq:dpdn2}
		\left[\frac{\p p}{\p n}\right]=&2[\mu]\left(\vec{n}\cdot\nabla_s^2\vec{q}+\kappa_t \frac{\p \vec{q}}{\p t}\cdot\vec{t}+\kappa_b \frac{\p \vec{q}}{\p b}\cdot\vec{b}\right) \nonumber \\
			& -\nabla_s\cdot\vec{f}+\left(\vec{f}\cdot\vec{n}\right)H+\left[\vec{g}\right]\cdot\vec{n},
	\end{align}
	where $\kappa_t$ and $\kappa_b$ are the principle curvatures along the $t$ and $b$ directions and recalling that $\vec{q}$ is the velocity on the interface.
\end{theorem}
\begin{proof}
	As Eq. (\ref{eq:momentum}) holds in each domain the jump in the system must be zero, $\left[\nabla p\right]=\left[\mu \nabla^2\vec{u}\right]+\left[\vec{g}\right]$.
	Dotting this by $\vec{n}$ results in Eq. (\ref{eq:dpdn1}). As has been noted elsewhere \cite{Ito2006} this form is not useful for numerical simulations due to the second-order
	partial derivatives of the velocity. 
	
	To obtain the alternative form begin by taking the surface divergence of the force balance on the interface, 
	$\nabla_s\cdot\left[\vec{\sigma}\cdot\vec{n}\right]=\nabla_s\cdot\vec{f}$, or in expanded form:
	\begin{equation}
		\label{eq:dpdn_der1}
		-\nabla_s\cdot\left[p\vec{n}\right]+\nabla_s\cdot\left[\mu\left(\nabla\vec{u}+\left(\nabla\vec{u}\right)^T\right)\cdot\vec{n}\right]=\nabla_s\cdot\vec{f}.
	\end{equation}
	Noting that the jump operator commutes with differentiation along the interface \cite{Ito2006, Xu2006}, individual portions of Eq. (\ref{eq:dpdn_der1})
	can be written as
	\begin{equation}
		\label{eq:dpdn_der2}
		\nabla_s\cdot\left[p\vec{n}\right]=\left[\nabla_s\cdot\left(p\vec{n}\right)\right] =\left[\nabla_sp\cdot\vec{n}+p\nabla_s\cdot\vec{n}\right]
			=[p]\nabla_s\cdot\vec{n}=[p]H=-\left(\vec{f}\cdot\vec{n}\right)H.	
	\end{equation}
	Next, write the surface divergence of a vector $\vec{v}$ on the interface as 
	\begin{equation}
		\nabla_s\cdot\vec{v}=\frac{\p \vec{v}}{\p t}\cdot\vec{t}+\frac{\p \vec{v}}{\p b}\cdot\vec{b}.
	\end{equation}
	Then,
	\begin{equation}
		\nabla_s\cdot\left[\mu\nabla\vec{u}\cdot\vec{n}\right]=\left[\mu\nabla_s\cdot\frac{\p \vec{u}}{\p n}\right]
			=\left[\mu \left(	\frac{\p}{\p t}\left(\frac{\p\vec{u}}{\p n}\right)\cdot\vec{t}	+\frac{\p}{\p b}\left(\frac{\p\vec{u}}{\p n}\right)\cdot\vec{b}		 \right)\right]			
	\end{equation}
	When considering the derivatives of $\p\vec{u}/\p n = \left(\nabla\vec{u}\right)\cdot\vec{n}$ along the tangent directions it is necessary to also take derivatives of the normal vector along the 
	same directions. This results in the following expression along the $t$-direction
	\begin{equation}
		\frac{\p}{\p t}\left(\frac{\p\vec{u}}{\p n}\right)\cdot\vec{t}=\frac{\p^2\vec{u}}{\p n\p t}\cdot\vec{t}+\left(\left(\nabla\vec{u}\right)^T\cdot\vec{t}\right)\cdot\frac{\p\vec{n}}{dt},	
	\end{equation}
	where the derivative along the $\vec{b}$-direction is similar. In the Darboux Frame it can be written that $\p \vec{n}/dt=\kappa_t \vec{t}+\tau_t\vec{b}$, where
	$\tau_t$ is the geodesic torsion along the $t$ direction. As the $t$-direction is a principle direction then the geodesic 
	torsion is zero, $\tau_t=0$ \cite{Carmo1976}. Using
	this results in the following expression,
	\begin{equation}
		\frac{\p}{\p t}\left(\frac{\p\vec{u}}{\p n}\right)\cdot\vec{t}=\frac{\p^2\vec{u}}{\p n\p t}\cdot\vec{t}+\kappa_t\frac{\p \vec{u}}{\p t}\cdot\vec{t}.
	\end{equation}
	Combining this with a similar result in the $b$-direction gives
	\begin{equation}
		\label{eq:dpdn_der3}
		\nabla_s\cdot \left[\mu\nabla\vec{u}\cdot\vec{n}\right]=
			\left[\mu\left(\frac{\p^2\vec{u}}{\p n\p t}\cdot\vec{t}+\kappa_t\frac{\p \vec{u}}{\p t}\cdot\vec{t} + \frac{\p^2\vec{u}}{\p n\p b}\cdot\vec{b}+\kappa_b\frac{\p \vec{u}}{\p b}\cdot\vec{b}		\right)\right]
	\end{equation}
	
	In a similar fashion the final portion of Eq. (\ref{eq:dpdn_der1}) can be written as 
	\begin{equation}
		\label{eq:dpdn_der4}
		\nabla_s\cdot \left[\mu\left(\nabla\vec{u}\right)^T\cdot\vec{n}\right]=
			\left[\mu\left(\frac{\p^2\vec{u}}{\p t^2}\cdot\vec{n}+\kappa_t\frac{\p \vec{u}}{\p t}\cdot\vec{t} + \frac{\p^2\vec{u}}{\p b^2}\cdot\vec{n}+\kappa_b\frac{\p \vec{u}}{\p b}\cdot\vec{b}		\right)\right]
	\end{equation}
	
	Combining Eqs. (\ref{eq:dpdn_der1}), (\ref{eq:dpdn_der2}), (\ref{eq:dpdn_der3}), and (\ref{eq:dpdn_der4}) and using Theorem \ref{thm:dudn2} and Theorem \ref{thm:slapdn} gives the expression
	\begin{equation}
		\label{eq:dpdn_der5}		
		\left[\mu\left(\frac{\p ^2\vec{u}}{\p n^2}\cdot\vec{n}\right)\right]=\left(\vec{f}\cdot\vec{n}\right) H-\nabla_s\cdot\vec{f} + 
					\left[\mu\left(\nabla_s^2\vec{u}\cdot\vec{n}
					+2\kappa_t\frac{\p \vec{u}}{\p t}\cdot\vec{t}+2\kappa_b\frac{\p \vec{u}}{\p b}\cdot\vec{b}\right)\right].
	\end{equation}
	
	Next, use Eq. (\ref{eq:split_lap}) to split the 
	Laplacian term in the base jump condition, Eq. (\ref{eq:dpdn1}), and combine with Eq. \ref{eq:dpdn_der5} to obtain
	\begin{equation}
		\label{eq:dpdn_der_6}
		\left[\frac{\p p}{\p n}\right]=
		2\left[\mu\left(\vec{n}\cdot\nabla_s^2\vec{u}+\kappa_t \frac{\p \vec{u}}{\p t}\cdot\vec{t}+\kappa_b \frac{\p \vec{u}}{\p b}\cdot\vec{b}\right)\right]
			 -\nabla_s\cdot\vec{f}+\left(\vec{f}\cdot\vec{n}\right)H+\left[\vec{g}\right]\cdot\vec{n}.
	\end{equation}
	
	Letting $\kappa_t^g$ be the geodesic curvature along the $t$ direction allows for the following to hold, 
	\begin{align}
		\label{eq:dpdn_der_7}
		\dfrac{\partial}{\partial t}\left[\dfrac{\partial\vec{u}}{\partial t}\right]&=\left[\dfrac{\partial}{\partial t}\left(\dfrac{\partial\vec{u}}{\partial t}\right)\right]=
			\left[\dfrac{\partial ^2\vec{u}}{\partial t^2}-\kappa_t^g\dfrac{\partial\vec{u}}{\partial b}-\kappa_t\dfrac{\partial\vec{u}}{\partial n}\right] \nonumber \\ 
			&=\left[\dfrac{\partial ^2\vec{u}}{\partial t^2}\right]-\kappa_t^g\left[\dfrac{\partial\vec{u}}{\partial b}\right]-\kappa_t\left[\dfrac{\partial\vec{u}}{\partial n}\right] \nonumber \\
			&=\left[\dfrac{\partial ^2\vec{u}}{\partial t^2}\right]-\kappa_t^g\dfrac{\partial\left[\vec{u}\right]}{\partial b}-\kappa_t\left[\dfrac{\partial\vec{u}}{\partial n}\right] \nonumber \\
			&=\left[\dfrac{\partial ^2\vec{u}}{\partial t^2}\right]-\kappa_t\left[\dfrac{\partial\vec{u}}{\partial n}\right].	
	\end{align}
	Due to the continuity of tangential derivatives, $[\partial \vec{u}/\partial t]=\p\left[\vec{u}\right]/\p t=0$ holds on the entire interface.
	It is thus possible to state that $\partial/\partial t [\partial \vec{u}/\partial t]=0$. Therefore
	\begin{equation}
		\label{eq:dpdn_der_8}
		\left[\dfrac{\partial ^2\vec{u}}{\partial t^2}\right]=\kappa_t\left[\dfrac{\partial\vec{u}}{\partial n}\right].
	\end{equation}

	When Eq. (\ref{eq:dpdn_der_8}) is dotted by the unit normal vector and noting 
	that $[\partial\vec{u}/\partial n]\cdot\vec{n}=0$, see Ref. \cite{Tan2009},
	the expression $\left[\partial^2\vec{u}/\partial t^2\right]\cdot\vec{n}=0$ holds. A similar expression holds in the $b$ direction.
	
	The exact form given in Eq. (\ref{eq:dpdn2}) can be obtained by writing 
	jumps of the form $[a\times b]$ as
	 $[a\times b]=[a]b^\pm+a^\mp[b]$.
	This allows for the following to hold,
	\begin{align}
		\left[\mu\frac{\p^2\vec{u}}{\p t^2}\cdot\vec{n}\right]=&[\mu]\frac{\p^2 \vec{q}}{\p t^2}\cdot\vec{n}+\mu^\mp\left[\frac{\p^2\vec{u}}{\p t^2}\right]\cdot\vec{n}
				=[\mu]\frac{\p^2 \vec{q}}{\p t^2}\cdot\vec{n}.
	\end{align}
	Replacing the appropriate terms in Eq. (\ref{eq:dpdn_der_6}) results in the alternative jump in the normal pressure derivative.	
\qquad \end{proof}

\begin{theorem}
	The jump in the second-normal derivative of the pressure is given by
	\begin{equation}
		\label{eq:pnn}
		\left[\frac{\p^2 p}{\p n^2}\right]=\nabla_s^2\left(\vec{f}\cdot\vec{n}\right)+\left[\nabla\cdot\vec{g}\right]-\left[\frac{\p p}{\p n}\right]H,
	\end{equation}
	where the jump in the first-normal derivative of the pressure is given in Eq. (\ref{eq:dpdn2}).
\end{theorem}
\begin{proof}
	As the pressure-Poisson equation, Eq. (\ref{eq:pressue-poisson}), holds in each domain the jump is given by $\left[\nabla^2 p\right]=\left[\nabla\cdot\vec{g}\right]$.
	Splitting the Laplacian of the pressure into the local frame results in
	\begin{equation}
		\left[\frac{\p^2 p}{\p n^2}+\nabla_s^2 p+\frac{\p p}{\p n}H\right]=\left[\nabla\cdot\vec{g}\right].
	\end{equation}
	Solving for the jump in the second derivative and using the result of Theorem \ref{thm:pressure_jump} results in Eq. (\ref{eq:pnn}).
\qquad \end{proof}

The final three theorems outline the derivation of the velocity jump conditions.

\begin{theorem}
	The jump in the augmented velocity is given by
	\begin{equation}
		\left[\mu\vec{u}\right]=\left[\mu\right]\vec{q}.
	\end{equation}
\end{theorem}
\begin{proof}
	Expanding the jump $\left[\mu\vec{u}\right]$ results in
	\begin{align}
		\left[\mu\vec{u}\right] = \left[\mu\right]\vec{q}+\mu^\pm\left[\vec{u}\right] = \left[\mu\right]\vec{q},
	\end{align}
	due to $[\vec{u}]=0$ and $\vec{u}^+=\vec{u}^-=\vec{q}$ on the interface.
\qquad \end{proof}

\begin{theorem}
	The jump in the normal derivative of the augmented velocity is given by
	\begin{equation}
		\label{eq:dudn}
		\left[\mu\dfrac{\partial\vec{u}}{\partial n}\right]=\vec{P}\vec{f}-[\mu]\left(\left(\dfrac{\partial\vec{q}}{\partial t}\cdot\vec{n}\right)\vec{t}+\left(\dfrac{\partial\vec{q}}{\partial b}\cdot\vec{n}\right)\vec{b}\right)
	\end{equation}
	where $\vec{P}$ is the projection operator given by $\vec{P}=\vec{I}-\vec{n}\otimes\vec{n}$.
\end{theorem}
\begin{proof}
	Dot the force balance on the interface, Eq. (\ref{eq:stress_balance}), by the tangent directions to obtain
	\begin{align}
		\left[\mu\dfrac{\partial\vec{u}}{\partial n}\right]\cdot\vec{t}=&\vec{f}\cdot\vec{t}-\left[\mu\dfrac{\partial\vec{u}}{\partial t}\right]\cdot\vec{n} \\
		\left[\mu\dfrac{\partial\vec{u}}{\partial n}\right]\cdot\vec{b}=&\vec{f}\cdot\vec{b}-\left[\mu\dfrac{\partial\vec{u}}{\partial b}\right]\cdot\vec{n}
	\end{align}
	In conjunction with Eq. (\ref{eq:normal_derivative}) a linear system can be written as 
	\begin{equation}
		\label{eq:jun}
		\begin{pmatrix}
			n_x & n_y & n_z \\
			t_x & t_y & t_z \\
			b_x & b_y & b_z
		\end{pmatrix}
		\left[\mu\dfrac{\partial\vec{u}}{\partial n}\right]=
		\begin{pmatrix}
			0 \\[15pt]
			\vec{f}\cdot\vec{t}-\left[\mu\dfrac{\partial\vec{u}}{\partial t}\right]\cdot\vec{n} \\[15pt]
			\vec{f}\cdot\vec{b}-\left[\mu\dfrac{\partial\vec{u}}{\partial b}\right]\cdot\vec{n}
		\end{pmatrix},
	\end{equation}
	where $\vec{n}=(n_x, n_y, n_z)^T$, $\vec{t}=(t_x, t_y, t_z)^T$, and $\vec{b}=(b_x, b_y, b_z)^T$ are the components of the local vectors.
	
	As the inverse of the matrix in Eq. (\ref{eq:jun}) is the transpose of the matrix the jump can be obtained as
	\begin{equation}
		\left[\mu\dfrac{\partial\vec{u}}{\partial n}\right]=\left(\vec{f}\cdot\vec{t}\right)\vec{t}+\left(\vec{f}\cdot\vec{b}\right)\vec{b}
				-[\mu]\left(\left(\dfrac{\partial\vec{q}}{\partial t}\cdot\vec{n}\right)\vec{t}+\left(\dfrac{\partial\vec{q}}{\partial b}\cdot\vec{n}\right)\vec{b}\right).
	\end{equation}
	Realizing that the term $\left(\vec{f}\cdot\vec{t}\right)\vec{t}+\left(\vec{f}\cdot\vec{b}\right)\vec{b}$ is simply the tangential projection of the force $\vec{f}$ onto the 
	interface the jump given in Eq. (\ref{eq:dudn}) is obtained.
\qquad \end{proof}

\begin{theorem}
	The jump in the second normal derivative of the augmented velocity is given by
	\begin{equation}
		\label{eq:dudn2}
		 \left[\mu\dfrac{\partial^2\vec{u}}{\partial n^2}\right]
				=-\left[\mu\right]\nabla_s^2\vec{q}-\left[\mu\dfrac{\partial \vec{u}}{\partial n}\right]H
				+\left[\dfrac{\partial p}{\partial n}\right]\vec{n} 
				-\nabla_s\left(\vec{f}\cdot\vec{n}\right)
				-\left[\vec{g}\right].
	\end{equation}
\end{theorem}
\begin{proof}
	Consider the jump of the augmented momentum balance equation, Eq. (\ref{eq:augmented_momentum}). Expand
	the Laplacian of the velocity and the gradient of the pressure into normal and tangential components:
	\begin{equation}
		\left[\mu\dfrac{\partial^2 \vec{u}}{\partial n^2}+\mu\nabla_s^2\vec{u}+\mu\dfrac{\partial \vec{u}}{\partial n}H\right]=\left[\dfrac{\partial p}{\partial n}\vec{n}+\dfrac{\partial p}{\partial b}\vec{b}+\dfrac{\partial p}{\partial t}\vec{t}\right]-\left[\vec{g}\right].
	\end{equation}
	Due to continuity of tangential derivatives, the unit normal, and the total curvature, it can be stated that 
	\begin{equation}
		\left[\mu\dfrac{\partial^2 \vec{u}}{\partial n^2}\right]=-\left[\mu\nabla_s^2\vec{u}\right]-\left[\mu\dfrac{\partial \vec{u}}{\partial n}\right]H
				+\left[\dfrac{\partial p}{\partial n}\right]\vec{n}
				+\dfrac{\partial [p]}{\partial b}\vec{b}+\dfrac{\partial [p]}{\partial t}\vec{t}-\left[\vec{g}\right].
	\end{equation}
	The jump in the pressure is simply $[p]=-\vec{f}\cdot\vec{n}$ while 
	\begin{equation}
		[\mu\nabla_s^2\vec{u}]=[\mu]\nabla_s^2\vec{q}+\mu^\pm[\nabla_s^2\vec{u}] 
								=[\mu]\nabla_s^2\vec{q}+\mu^\pm\nabla_s^2[\vec{u}]
								=[\mu]\nabla_s^2\vec{q}.
	\end{equation}
	Finally, noticing that $(\partial/\partial t)\vec{t}+(\partial/\partial b)\vec{b}=\nabla_s$ is the surface gradient the jump condition given
	in Eq. (\ref{eq:dudn2}) can be obtained.
	
\qquad \end{proof}

\section{\label{sec:TestCases} Sample Test Case}

To verify the accuracy of the derived jump conditions a sample analytic test case has been created.
To create this test case valid inner and outer velocity fields, $\vec{u}^-$ and $\vec{u}^+$, are determined
for a pre-determined interface. These velocity fields are chosen to enforce the conditions $\vec{u}^-=\vec{u}^+$ and 
$\nabla_s\cdot\vec{u}^-=\nabla_s\cdot\vec{u}^+=0$ on the interface $\Gamma$ and $\nabla\cdot\vec{u}^-=\nabla\cdot\vec{u}^+=0$ 
in the domain $\Omega$. Next, arbitrary inner and outer pressure fields, $p^-$ and $p^+$ are chosen. 
Using the velocity and pressure field the inner and outer body forces, $\vec{g}^-$ and $\vec{g}^+$ can be calculated
using Eq. (\ref{eq:momentum}). The singular force, $\vec{f}$ is obtained through the force balance on the interface, Eq. (\ref{eq:stress_balance}).

The analytic body and singular forces are then used in conjunction with the known interface velocity $\vec{q}$ and 
appropriate boundary conditions to solve the discontinuous Stokes equations
in the domain using the Immersed Interface Method. The numerical pressure and velocity fields are compared to their analytic versions to determine the 
overall spatial accuracy of the method.

%\subsection{Spherical Interface}

In this example the interface is a unit sphere centered at the origin. The outer viscosity is taken to be 1 while the inner viscosity is given by $\mu$.
The exact inner velocity and pressure for the region $x^2+y^2+z^2\leq 1$ are given by
\begin{align}
	\vec{u}^-=&\left(\begin{array}{c}
		2yz \\
		-xz \\
		-xy \\
		\end{array}\right), \\
	p^-=&\cos(x)\sin(y+z).
\end{align}
This results in an inner body force of the form
\begin{equation}
	\vec{g}^-=\left(\begin{array}{c}
		-\sin(x)\sin(y+z) \\
		 \cos(x)\cos(y+z)\\
		 \cos(x)\cos(y+z)\\
		\end{array}\right). \\
\end{equation}

The exact outer velocity and pressure for the region $x^2+y^2+z^2>1$ is
\begin{align}
	\vec{u}^+=&\frac{1}{x^2+y^2+z^2}\left(\begin{array}{c}
		2yz \\
		-xz \\
		-xy \\
		\end{array}\right), \\
	p^+=&0,
\end{align}
which results in a body force of 
\begin{equation}
	\vec{g}^+=\frac{1}{x^2+y^2+z^2}\left(\begin{array}{c}
		12 yz \\
		-6xz\\
		-6xy\\
		\end{array}\right). \\
\end{equation}

A force balance on the interface results in a singular force of
\begin{equation}
	\vec{f}=\left(\begin{array}{c}
		-2(1+\mu)y z+x\cos(x)\sin(y+z) \\
		 (1+\mu)x z+y\cos(x)\sin(y+z)\\
		 (1+\mu)x y+z\cos(x)\sin(y+z)\\
		\end{array}\right). \\
\end{equation}

Clearly on the interface $\vec{q}=\vec{u}^-=\vec{u}^+$ holds on $\Gamma$ and can be applied in the derived jump conditions.

%\subsection{Ellipsoidal Interface}

%In this example the interface is a prolate ellipsoidal interface centered at the origin. 
%The ellipsoid has axes of $(1,1.5,1)$ in the $(x,y,z)$-directions.
%The outer viscosity is taken to be 1 while the inner viscosity is given by $\mu$.
%The exact inner velocity and pressure for the region $x^2+y^2+z^2\leq 1$ are given by
%\begin{align}
	%\vec{u}^-=&\left(\begin{array}{c}
		%y \\
		%-x \\
		%0 \\
		%\end{array}\right), \\
	%p^-=&\cos(x)\sin(y+z).
%\end{align}
%This results in an inner body force of the form
%\begin{equation}
	%\vec{g}^-=\left(\begin{array}{c}
		%-\sin(x)\sin(y+z) \\
		 %\cos(x)\cos(y+z)\\
		 %\cos(x)\cos(y+1)\\
		%\end{array}\right). \\
%\end{equation}

%The exact outer velocity and pressure for the region $x^2+y^2+z^2>1$ is
%\begin{align}
	%\vec{u}^+=&\frac{1}{810\left(x^2+y^2+z^2\right)}\times \nonumber \\ 
			%&\left(\begin{array}{c}
		%y\left(9+72x^2+32y^2+72z^2\right)\\
		%-9x\left(9+18x^2+8y^2+18z^2\right)\\
		%0
		%\end{array}\right), \\
	%p^+=&0,
%\end{align}
%which results in a body force of 
%\begin{equation}
	%\vec{g}^-=\frac{1}{270\left(x^2+y^2+z^2\right)}\left(\begin{array}{c}
		%-160y\\
		%480x\\		
		%0
		%\end{array}\right). \\
%\end{equation}

%A force balance on the interface results in a singular force of
%\begin{align}
	%\vec{f}=&\frac{1}{405\sqrt{81-20y^2}}\times \nonumber \\
			%&\left(\begin{array}{c}
		%8y\left(81-20y^2\right)+3645x\cos(x)\sin(y+z)\\
		%18\left(x\left(-81+20y^2\right)+90 y\cos(x)\sin(y+z)\right)\\
		%3645z\cos(x)\sin(y+z)\\
		%\end{array}\right).
%\end{align}

\section{\label{sec:Results} Results}

In the following results the domain is a taken to be a $[-2,2]^3$ cube. The pressure and velocity fields are obtained using a split calculation where
the pressure is calculated first by solving Eq. (\ref{eq:pressue-poisson}) and then obtaining the velocity field through the solution
of Eq. (\ref{eq:momentum}). It should be noted that the jump conditions allow for any Stokes field solution technique, not only the one used here.

Sample pressure results for the $z=0$ plane on a $129^3$ grid are shown in Fig. (\ref{fig:SamplePressure}). 
The external viscosities shown are $0.001$ and $1000$. 
It is clear that the sharp discontinuity of the pressure is captured
by the method and the result is not greatly affected by the viscosity ratio.
\begin{figure}[ht]
	\begin{center}
		\subfigure[$\mu^+=0.001$]{
			\includegraphics[width=0.65\textwidth]{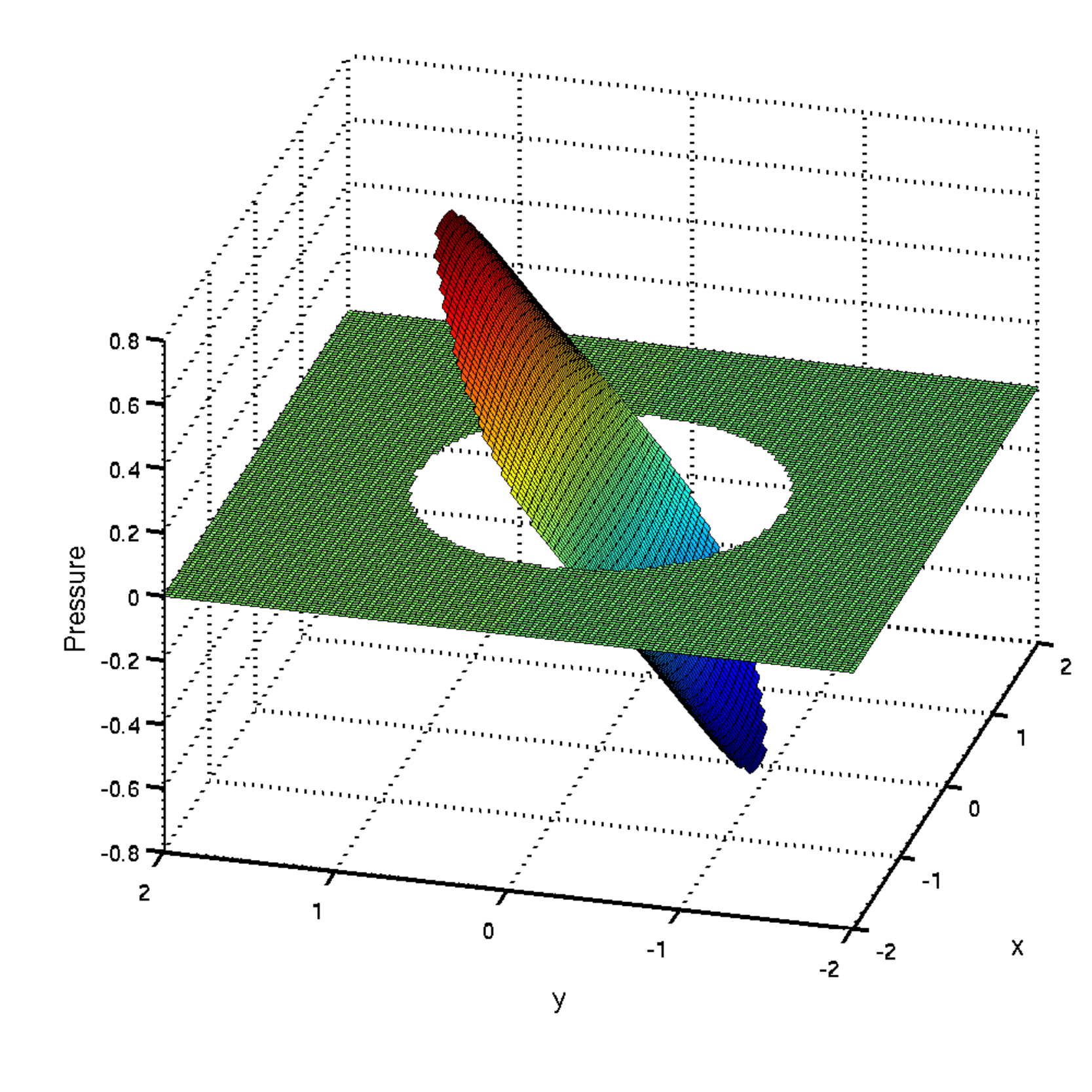}
		} \\
		\subfigure[$\mu^+=1000$]{
			\includegraphics[width=0.65\textwidth]{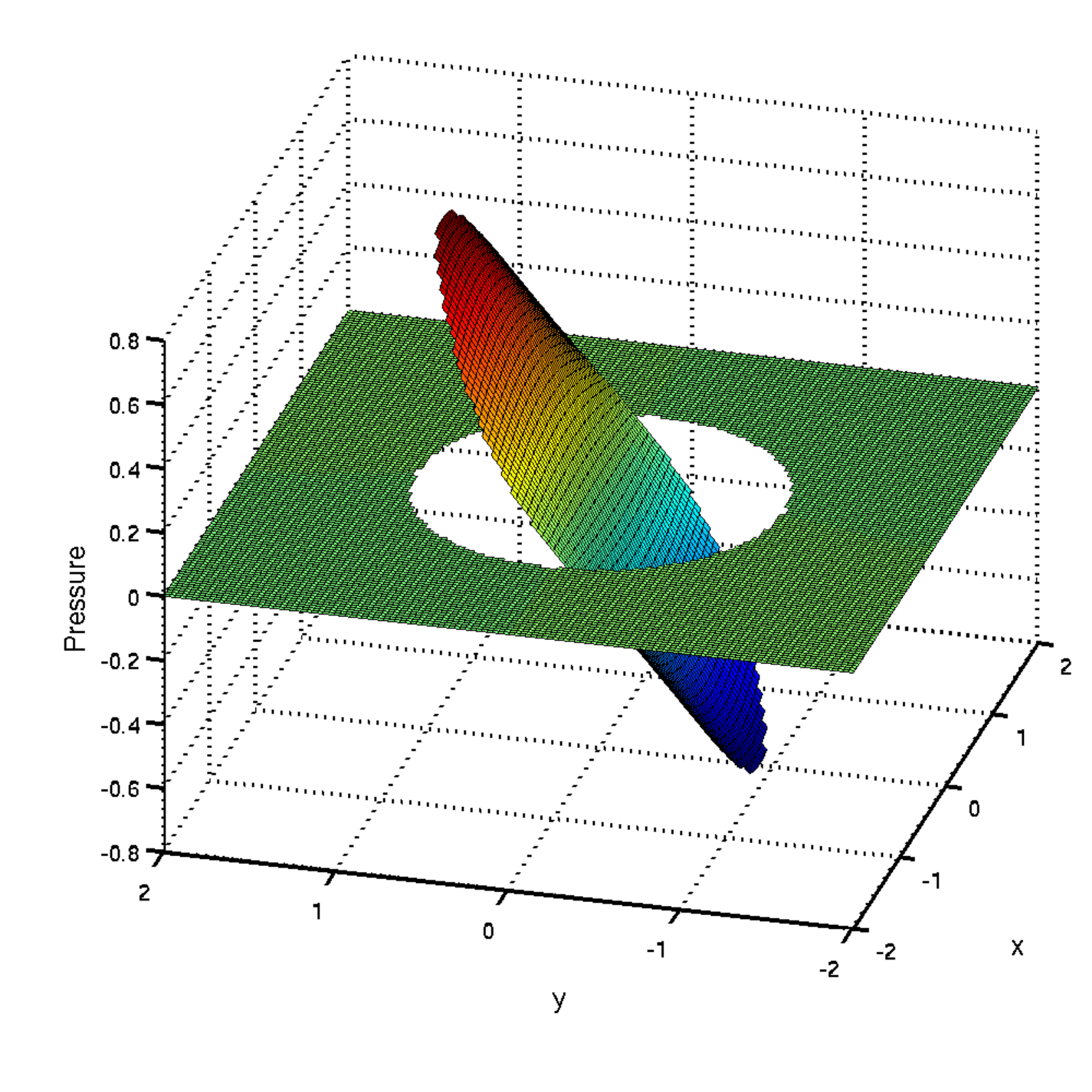}
		}
	
		\caption{\label{fig:SamplePressure} Example of the pressure distribution on the $z=0$ plane for a $129^3$ grid for outer viscositoes of 
		$0.001$ and $1000$. The inner viscosity is equal to 1.}
	\end{center}	
\end{figure}	
\clearpage

To demonstrate that the derived jump conditions, when used with a second-order Immersed Interface Method, achieve the desired accuracy the errors
for both the pressure and velocity fields are now presented. In calculating the error external viscosities ranging from $10^{-3}$ to $10^3$ have been
considered. Grid sizes ranging from $48^3$ to $384^3$ have been used, corresponding to grid spacings of $h\approx 0.0115$ to $h\approx 0.0851$.

\begin{figure}[ht]
	\begin{center}
		\subfigure[]{
			\includegraphics[width=0.65\textwidth]{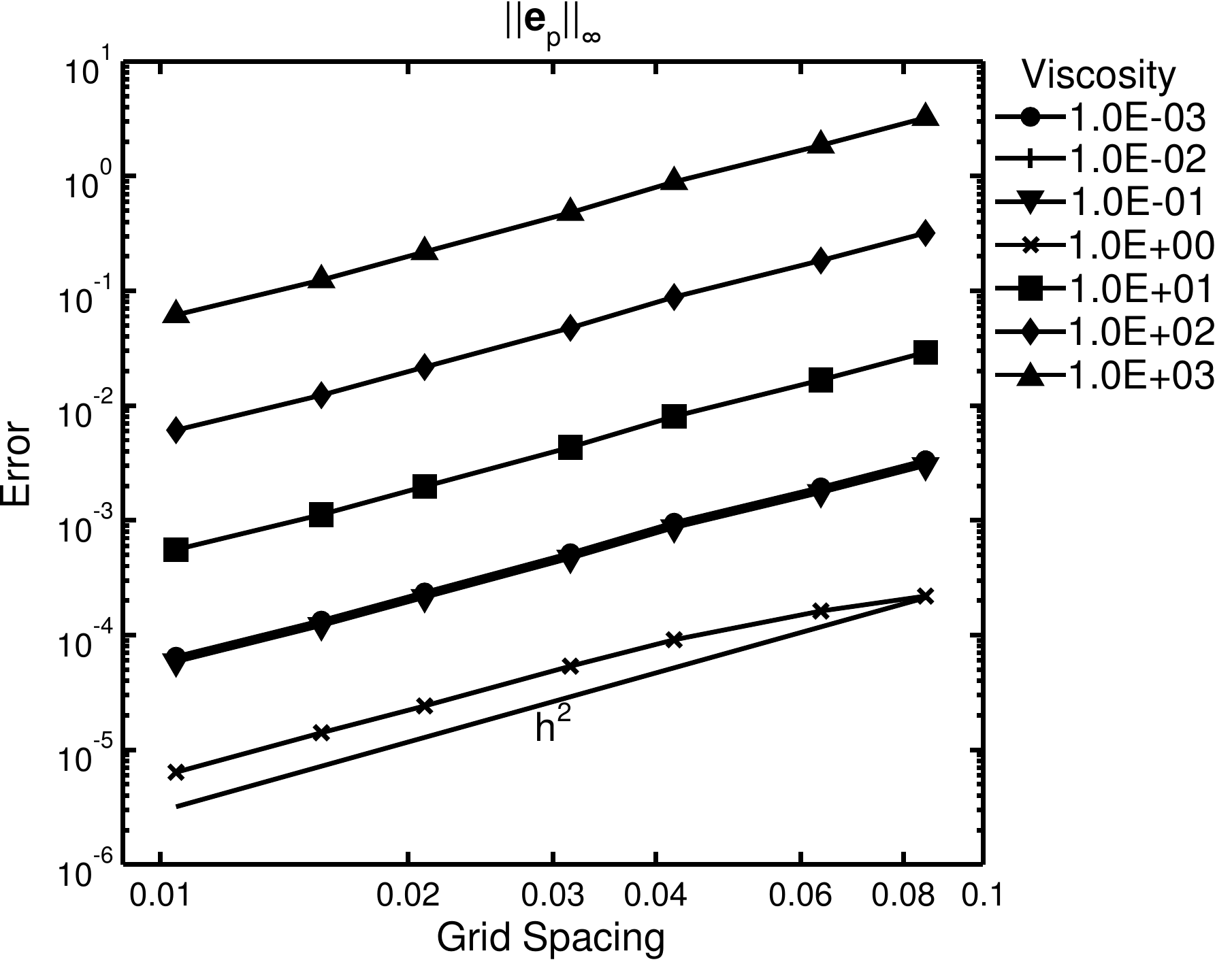}
		} \\
		\subfigure[]{
			\includegraphics[width=0.65\textwidth]{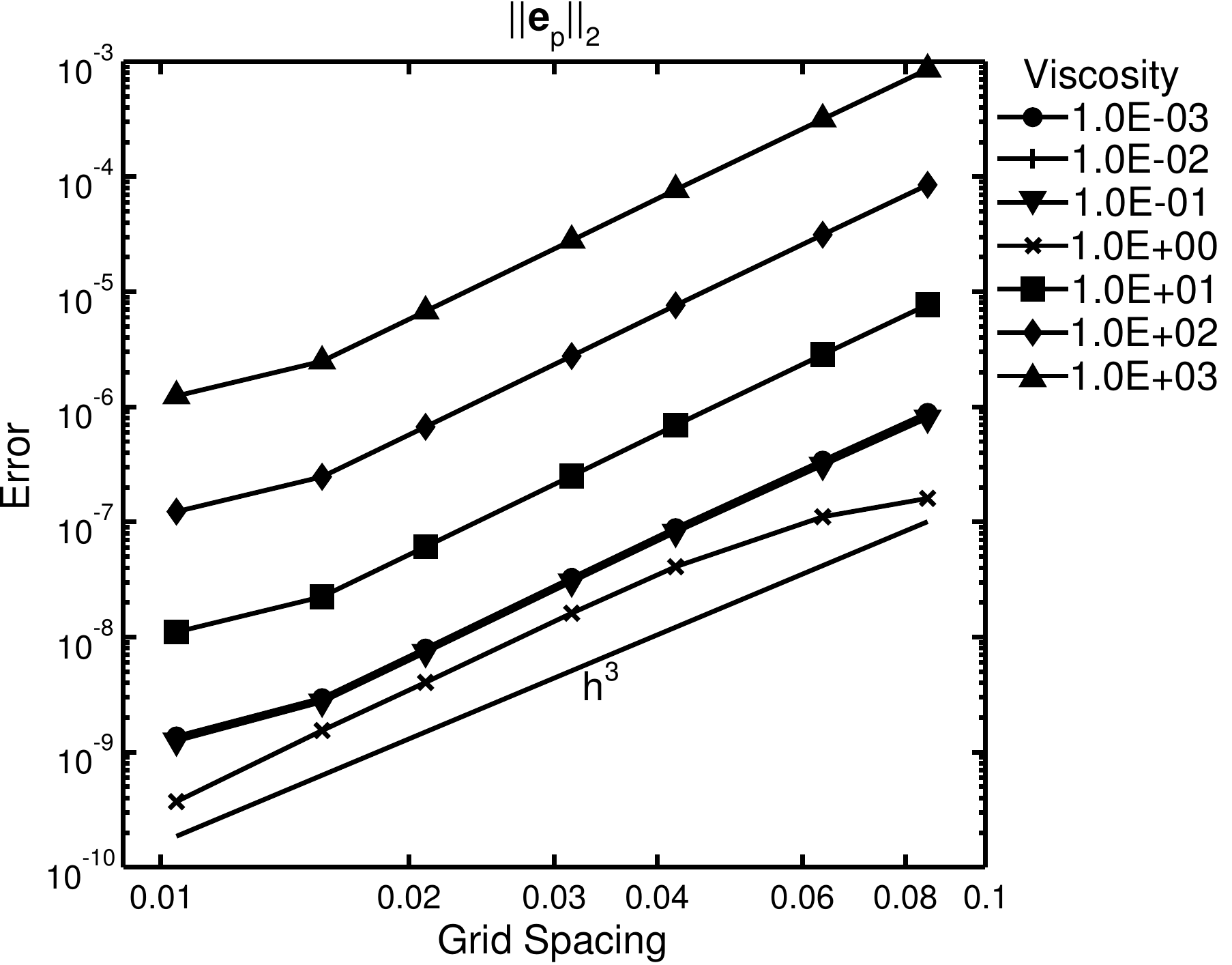}
		}
	
		\caption{\label{fig:PressureConvergence} Error of the pressure field versus grid spacing. The shown viscosities are the outer viscosities, $\mu^+$.
		The inner viscosity is taken to be $\mu^-=1$.}
	\end{center}	
\end{figure}	

\begin{figure}[ht]
	\begin{center}
		\subfigure[]{
			\includegraphics[width=0.65\textwidth]{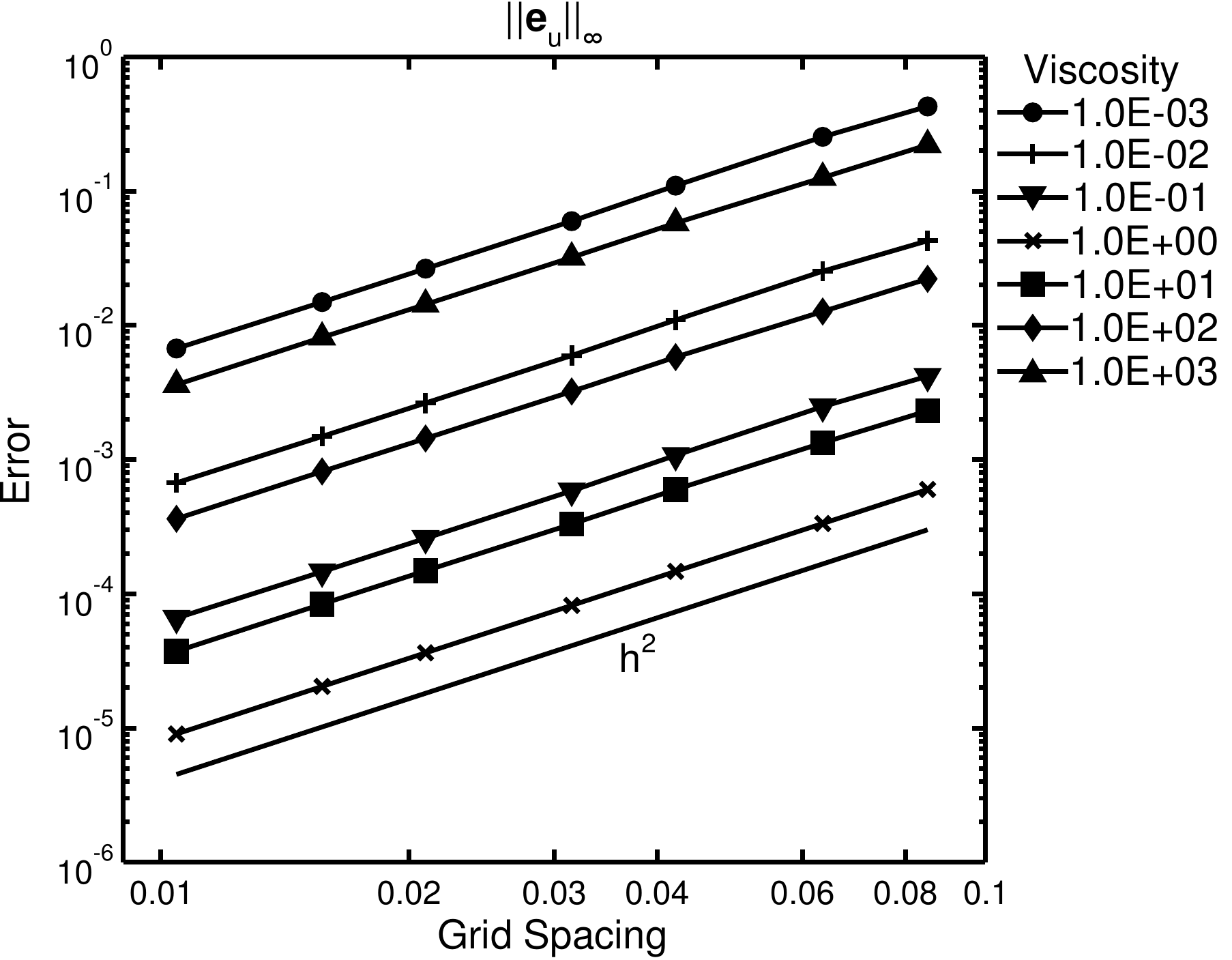}
		} \\
		\subfigure[]{
			\includegraphics[width=0.65\textwidth]{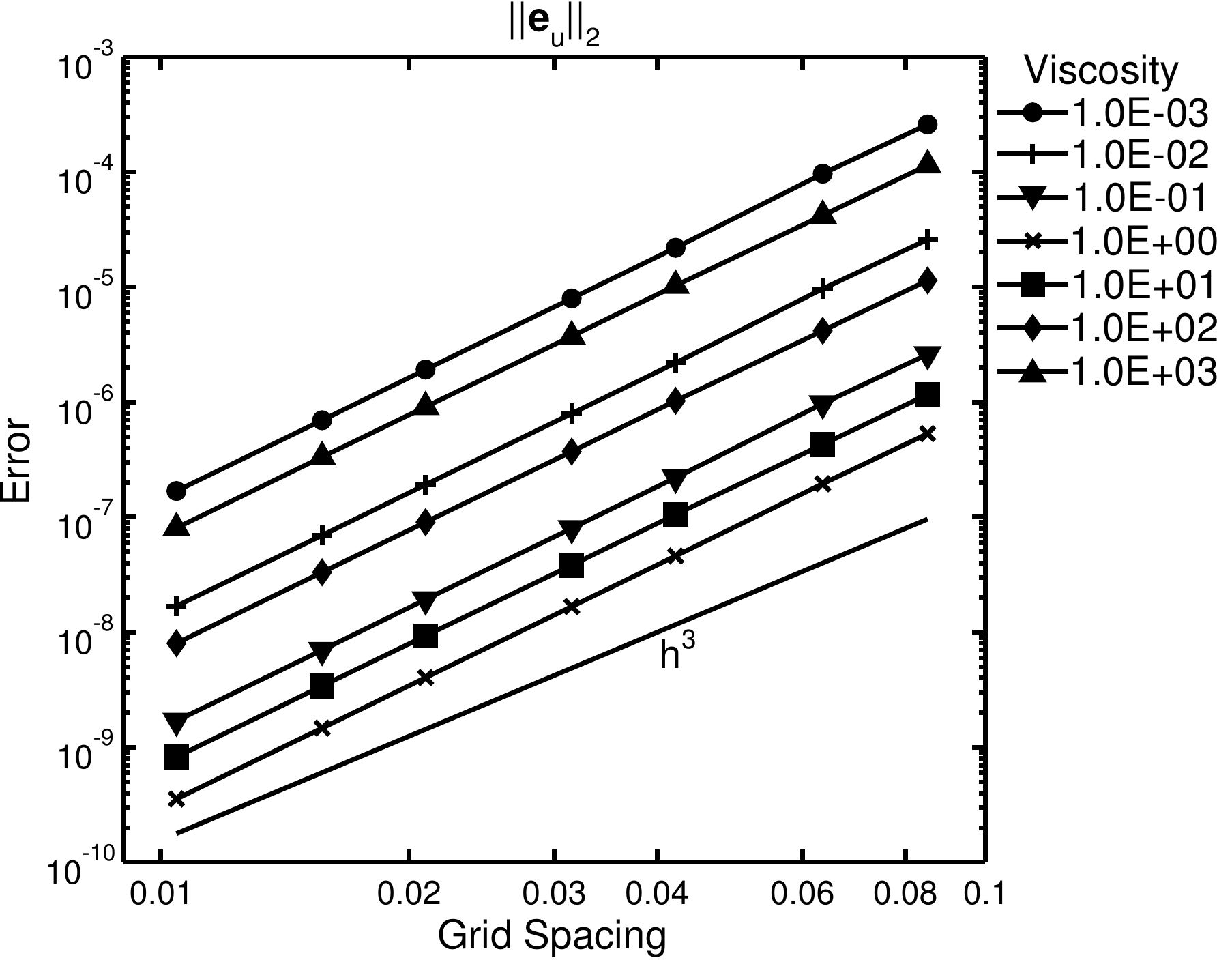}
		}
	
		\caption{\label{fig:VelocityConvergence} Error of the velocity field versus grid spacing. The shown viscosities are the outer viscosities, $\mu^+$.
		The inner viscosity is taken to be $\mu^-=1$.}
	\end{center}	
\end{figure}
\clearpage

The error for the pressure field are presented in Fig. \ref{fig:PressureConvergence} while the corresponding velocity errors are shown in 
Fig. \ref{fig:VelocityConvergence}. For both the pressure and velocity fields the second-order accuracy of the underlying discretization
is clearly preserved despite the presence of the discontinuities in the solution field. It should be noted that the overall accuracy 
of the method depends on the viscosity difference between the inner and outer fluid, with matched viscosity producing the lowest error. This 
is to be expected as the jump conditions are heavily dependent on this viscosity difference. Discretization errors in the surface derivative 
quantities, such as the surface Laplacian of Eq. (\ref{eq:dpdn2}), are amplified as the viscosity difference increases. While 
more accurate surface derivative calculations would reduce this dependence, it will always remain and must be taken into account.

\section{\label{sec:Conclusion} Concluding Remarks}

In this paper the three-dimensional, normal jump conditions for the pressure and velocity fields have been derived for an embedded incompressible/inextensible interface 
with a singular force and in the presence of a discontinuous viscosity field. The jump conditions take into account the additional constraint
of interface incompressibility given by a surface-divergent free velocity field. The jump conditions are applicable to any 
Stokes solver using the Immersed Interface Method. As a demonstration of the accuracy of the jump conditions a sample analytic test case
has been created. Error analysis using this test case demonstrates that a second-order accurate discretization maintains the expected accuracy
despite the discontinuity of the solution fields. 
Future work will use these jump conditions to construct generalized Immersed Interface Methods which are capable of enforcing 
surface incompressibility while solve the Stokes equations for multiphase Stokes flow with discontinuous viscosity.

\bibliographystyle{siam}
\bibliography{Vesicle}

\begin{thebibliography}{10}

\bibitem{Adams2002}
{\sc Loyce Adams and Zhilin Li}, {\em {The Immersed Interface/Multigrid Methods
  for Interface Problems}}, SIAM J. Sci. Comput., 24 (2002), pp.~463--479.

\bibitem{Biben2005}
{\sc T~Biben, K~Kassner, and C~Misbah}, {\em {Phase-field approach to
  three-dimensional vesicle dynamics}}, Phys. Rev. E, 72 (2005), p.~041921.

\bibitem{Carmo1976}
{\sc MP~Carmo}, {\em {Differential geometry of curves and surfaces}},
  Prentice-Hall, 1976.

\bibitem{Chang1996}
{\sc YC~Chang, TY~Hou, B~Merriman, and S~Osher}, {\em {A level set formulation
  of eulerian interface capturing methods for incompressible fluid flows}}, J.
  Comput. Phys., 124 (1996), pp.~449--464.

\bibitem{Ito2006}
{\sc K~Ito, Z~Li, and X~Wan}, {\em {Pressure Jump Conditions for Stokes
  Equations with Discontinuous Viscosity in 2D and 3D}}, Methods Appl. Anal.,
  13 (2006), pp.~199--214.

\bibitem{Lai200899}
{\sc M.C. Lai and H.C. Tseng}, {\em {A simple implementation of the immersed
  interface methods for Stokes flows with singular forces}}, Comput. \& Fluids,
  37 (2008), pp.~99--106.

\bibitem{Le2006109}
{\sc D.V. Le, B.C. Khoo, and J.~Peraire}, {\em {An immersed interface method
  for viscous incompressible flows involving rigid and flexible boundaries}},
  J. Comput. Phys., 220 (2006), pp.~109--138.

\bibitem{Leveque1997}
{\sc RJ~Leveque and ZL~Li}, {\em {Immersed interface methods for Stokes flow
  with elastic boundaries or surface tension}}, SIAM J. Sci. Comput., 18
  (1997), pp.~709--735.

\bibitem{Leveque1994}
{\sc Randall~J Leveque and Zhilin Li}, {\em {The immersed interface method for
  elliptic equations with discontinuous coefficients and singular sources}},
  SIAM J. Numer. Anal., 31 (1994), pp.~1019--1044.

\bibitem{Li2001}
{\sc Z.~Li and K.~Ito}, {\em {Maximum Principle Preserving Schemes for
  Interface Problems with Discontinuous Coefficients}}, SIAM J. Sci. Comput.,
  23 (2001), pp.~339--361.

\bibitem{li2006immersed}
\leavevmode\vrule height 2pt depth -1.6pt width 23pt, {\em {The Immersed
  Interface Method: Numerical Solutions of PDEs Involving Interfaces and
  Irregular Domains}}, {Frontiers in Applied Mathematics}, Society for
  Industrial and Applied Mathematics, 2006.

\bibitem{Li2001822}
{\sc Zhilin Li and Ming-Chih Lai}, {\em {The Immersed Interface Method for the
  Navier--Stokes Equations with Singular Forces}}, J. Comput. Phys., 171
  (2001), pp.~822--842.

\bibitem{Lowengrub2009}
{\sc JS~Lowengrub, A~Ratz, and A~Voigt}, {\em {Phase-field modeling of the
  dynamics of multicomponent vesicles: Spinodal decomposition, coarsening,
  budding, and fission}}, Phys. Rev. E, 79 (2009), p.~031926.

\bibitem{Mayo1984}
{\sc Anita Mayo}, {\em {The Fast Solution of Poisson's and the Biharmonic
  Equations on Irregular Regions}}, SIAM J. Numer. Anal., 21 (1984),
  pp.~285--299.

\bibitem{Peskin1977}
{\sc Charles~S Peskin}, {\em {Numerical analysis of blood flow in the heart}},
  J. Comput. Phys., 25 (1977), pp.~220--252.

\bibitem{Pozrikidis2003}
{\sc C~Pozrikidis}, {\em {Modeling and simulation of capsules and biological
  cells}}, CRC Press, 2003, ch.~Shell theory for capsules and shells.,
  pp.~35--102.

\bibitem{Russell2003}
{\sc D~Russell and ZJ~Wang}, {\em {A cartesian grid method for modeling
  multiple moving objects in 2D incompressible viscous flow}}, J. Comput.
  Phys., 191 (2003), pp.~177--205.

\bibitem{Salac2011}
{\sc D.~Salac and M.~Miksis}, {\em {A level set projection model of lipid
  vesicles in general flows}}, J. Comput. Phys., 230 (2011), pp.~8192--8215.

\bibitem{Salac2012b}
{\sc David Salac and Michael~J Miksis}, {\em {Reynolds number effects on lipid
  vesicles}}, J. Fluid Mech., 711 (2012), pp.~122--146.

\bibitem{Seifert1997}
{\sc U~Seifert}, {\em {Configurations of fluid membranes and vesicles}}, Adv.
  Phys., 46 (1997), pp.~13--137.

\bibitem{Tan2012}
{\sc Z~Tan, D.V. Le, K.M Kim, and B.C. Khoo}, {\em {An Immersed Interface
  Method for the Simulation of Inextensible Interfaces in Viscous Fluids }},
  Commun. Comput. Phys., 11 (2012), pp.~925--950.

\bibitem{Tan20089955}
{\sc Z.~Tan, D.V. Le, Z.~Li, K.M. Lim, and B.C. Khoo}, {\em {An immersed
  interface method for solving incompressible viscous flows with piecewise
  constant viscosity across a moving elastic membrane}}, J. Comput. Phys., 227
  (2008), pp.~9955--9983.

\bibitem{Tan2009}
{\sc Z~Tan, D~Le, K~Lim, and B~Khoo}, {\em {An Immersed Interface Method for
  the Incompressible Navier--Stokes Equations with Discontinuous Viscosity
  Across the Interface}}, SIAM J. Sci. Comput., 31 (2009), pp.~1798--1819.

\bibitem{Veerapaneni2009}
{\sc SK~Veerapaneni, D~Gueyffier, D~Zorin, and G~Biros}, {\em {A boundary
  integral method for simulating the dynamics of inextensible vesicles
  suspended in a viscous fluid in 2D}}, J. Comput. Phys., 228 (2009),
  pp.~2334--2353.

\bibitem{Veerapaneni2011}
{\sc SK~Veerapaneni, A~Rahimian, G~Biros, and D~Zorin}, {\em {A fast algorithm
  for simulating vesicle flows in three dimensions}}, J. Comput. Phys., 230
  (2011), pp.~5610--5634.

\bibitem{Xu2003}
{\sc Jian-Jun Xu and Hong-Kai Zhao}, {\em {An Eulerian Formulation for Solving
  Partial Differential Equations Along a Moving Interface}}, J. Sci. Comput.,
  19 (2003), pp.~573--594.

\bibitem{Xu2006454}
{\sc Sheng Xu and Z.~Jane Wang}, {\em {An immersed interface method for
  simulating the interaction of a fluid with moving boundaries}}, J. Comput.
  Phys., 216 (2006), pp.~454--493.

\bibitem{Xu2006}
{\sc Sheng Xu and Z~Jane Wang}, {\em {Systematic derivation of jump conditions
  for the immersed interface method in three-dimensional flow simulation}},
  SIAM J. Sci. Comput., 27 (2006), pp.~1948--1980.

\end{thebibliography}

\end{document}